\documentclass[11pt]{amsart}
\usepackage{amsmath, amssymb, amsbsy, amsfonts, amsthm, latexsym, amsopn, amstext, amsxtra, euscript, amscd, color, mathrsfs}
\usepackage[normalem]{ulem}
\usepackage{soul}

\PassOptionsToPackage{hyphens}{url}\usepackage{hyperref}
\makeatletter
\@namedef{subjclassname@2020}{%
  \textup{2020} Mathematics Subject Classification}
\makeatother
 
 \usepackage[capbesideposition=outside,capbesidesep=quad]{floatrow}
\usepackage{longtable}
\usepackage{tabularx}

\restylefloat{table}
\restylefloat{table}
            
\usepackage{multirow,caption}
\usepackage{array}
\usepackage{amscd}
\usepackage{color,enumerate}
\newcommand{\RNum}[1]{\lowercase\expandafter{\romannumeral #1\relax}}

\usepackage[colorinlistoftodos,prependcaption,textsize=tiny]{todonotes}

\newtheorem{thm}{Theorem}[section]
\newtheorem{lem}[thm]{Lemma}
\newtheorem{cor}[thm]{Corollary}

\newtheorem{exmp}[thm]{Example}

\newtheorem{rem}[thm]{Remark}

\newtheorem{rmk}[thm]{Remark}

\newtheorem{thm-con}[thm]{Theorem-Conjecture}

\theoremstyle{definition}
\newtheorem{defn}[thm]{Definition}

\def\rank{{\rm rank}}

\newcommand{\F}{\mathbb F}

\def\Tr{{\rm Tr}}

\begin{document}
\title[Second-order zero differential spectra of certain maps]{The second-order zero differential spectra of some APN and other maps over finite fields}
 
 \author[K. Garg]{Kirpa Garg}
 \address{Department of Mathematics, Indian Institute of Technology Jammu, Jammu 181221, India}
 \email{kirpa.garg@gmail.com}
 
 \author[S. U. Hasan]{Sartaj Ul Hasan}
 \address{Department of Mathematics, Indian Institute of Technology Jammu, Jammu 181221, India}
 \email{sartaj.hasan@iitjammu.ac.in}
 
 \author[C. Riera]{Constanza Riera}
 \address{Department of Computer Science, Electrical Engineering and Mathematical Sciences, Western
 Norway University of Applied Sciences, 5020 Bergen, Norway}
 \email{csr@hvl.no}
 
  \author[P.~St\u anic\u a]{Pantelimon~St\u anic\u a}
  \address{Applied Mathematics Department, Naval Postgraduate School, Monterey, CA 93943, USA}
 \email{pstanica@nps.edu}

\keywords{Finite fields, almost perfect non-linear functions, second-order zero differential spectra, second-order zero differential uniformity}

\subjclass[2020]{12E20, 11T06, 94A60}

\begin{abstract}
The Feistel Boomerang Connectivity Table and the related notion of $F$-Boomerang uniformity (also known as the second-order zero differential uniformity) has been recently introduced by Boukerrou et al.~\cite{Bouk}. These tools shall provide a major impetus in the analysis of the security of the Feistel network-based ciphers. In the same paper, a characterization of almost perfect nonlinear functions (APN) over fields of even characteristic in terms of second-order zero differential uniformity was also given. Here, we find a sufficient condition for an odd or even function over fields of odd characteristic to be an APN function, in terms of second-order zero differential uniformity. Moreover, we compute the second-order zero differential spectra of several APN or other low differential uniform functions, and show that our considered functions also have  low second-order zero differential uniformity, though it may  vary widely, unlike the case for even characteristic when it is always zero.
\end{abstract}

\maketitle

\section{Introduction}

In symmetric-key cryptography, the substitution box (S-box) plays a crucial role in most of modern block ciphers. There are many cryptographic attacks that are possible on these block ciphers. One of the most powerful attacks introduced by Biham and Shamir~\cite{Biham91} in 1991 is known as differential cryptanalysis. Further, to quantify the degree of resistance of these functions against differential attacks, Nyberg~\cite{Nyberg} introduced the notion of differential uniformity. The smaller the differential uniformity of a function, the stronger its resistance to the differential attack. 

Let $\F_{q}$ denotes the finite field with $q=p^n$ elements where $n$ is a positive integer and $p$ is a prime number. We use the notation $\F_{q}^{*}$ to represent the multiplicative cyclic group of non-zero elements of $\F_{q}$. Let $F$ be a function from $\F_{q}$ to $\F_{q}$, and $a$ be any element in $\F_{q}$, then we define the derivative of $F$ in the direction $a$ as $D_F(X,a) := F(X+a)-F(X)$ for all $X\in \F_{q}$. The Difference Distribution Table (DDT) entry $\Delta_F(a,b)$ for any $a, b \in \F_q$, at point $(a, b)$ is given by the number of solutions $X\in \F_{q}$ satisfying $D_F(X,a) = b$.
The differential uniformity of $F$, denoted by $\Delta_F$, is given by $\Delta_{F} := \max\{\Delta_{F}(a, b) : a\in \F_{q}^*, b \in \F_{q} \}.$
We say that the function $F$ is a perfect nonlinear (PN) function if $\Delta_{F} = 1$ and an almost perfect nonlinear (APN) function if $\Delta_{F} = 2$.

The boomerang attack is another significant cryptanalysis technique introduced by Wagner~\cite{DW} for analyzing block ciphers. This attack can be seen as a variant of the classical differential attack. At Eurocrypt 2018, Cid et al.~\cite{cid} analyze the boomerang style attack in a better way by introducing a new tool known as the Boomerang Connectivity Table
(BCT). To quantify the resistance of a function against the boomerang attack,
Boura and Canteaut~\cite{BoCa} introduced the concept of boomerang uniformity. However, this was valid only for Substitution-Permutation Network (SPN) ciphers and the obvious question was to define its equivalent for Feistel ciphers. It is important to note that Feistel network-based ciphers, such as 3-DES and CLEFIA~\cite{SSA}, are equally significant in block cipher design as SPN. Boukerrou, Huynh, Lallemand, Mandal, and Minier~\cite{Bouk} addressed the requirement for a counterpart to the BCT, and extended this concept to Feistel ciphers. They introduced the Feistel Boomerang Connectivity Table (FBCT) as an extension for Feistel ciphers, where the S-boxes may not be permutations. Further, these coefficients of the FBCT are related to the second-order zero differential spectra of the functions over finite fields. 

In~\cite{Bouk}, the authors studied the FBCT entries for some APN and inverse functions over finite fields of even characteristics. Li et al.~\cite{LYT} extended this research to odd characteristics, and studied the second-order zero differential spectra of some APN power functions over odd characteristic. Furthermore, Eddahmani and Mesnager~\cite{EM} investigated the entries of the  FBCT for the inverse, Gold, and Bracken-Leander functions over finite fields with even characteristic. Recently, Man et al.~\cite{MMN} also computed the FBCT entries for a specific power function, $F(X)=X^{2^{m+1}-1}$ over $\F_{2^{n}}$, where $n=2m$ or $n=2m+1$. Further, Garg et al.~\cite{GHRS} also considered several APN or other low differential uniform functions in both even and odd characteristics and investigated their second-order zero differential spectra.

It was shown in~\cite{Bouk} that for the finite field of even characteristic, all the non-trivial coefficients of FBCT of $F$ are 0 if and only if $F$ is APN, or equivalently $F$ is second-order zero differentially $0$-uniform if and only if it is APN.  In this paper, we show that, for odd characteristic, if an odd or even function $F$ is second-order zero differentially $1$-uniform then it has to be  an APN function. The converse, however, does not hold. We further  compute the second-order zero differential spectra of some other functions that are APN, or have other low differential uniformity over finite fields of both even and odd characteristic.   The paper is structured as follows. In Section~\ref{S2}, we recall some definitions and results required later. In Section~\ref{S3}, we give the only known constraint with respect to the second-order zero differential uniformity, for functions that are not PN in odd characteristic, namely, we show that  an odd or even function $F$, which is second-order zero differentially $1$-uniform must be APN. We also consider the second-order zero differential spectra of some functions (including APNs) over finite fields of odd characteristic. Furthermore, Section~\ref{S4} explores the second-order zero differential spectrum of some (partial) $0$-APN functions and extensions. We conclude the paper in Section~\ref{comments}.

\section{Preliminaries}\label{S2}

In this section, we recall some definitions and results that we will use in subsequent sections.

\begin{defn} For $p$ an odd prime, $n$ a positive integer, and $q=p^n$,  we let $\chi$ be the 
quadratic character of $\mathbb F_{q}$   defined by
$$
 \chi (X):= \begin{cases}  
0 &\quad \text{if $X=0$},\\
1 &\quad \text{if $X$ is a square of an element of $\mathbb F_{q}^*$}, \\ 
 -1 &\quad \text{otherwise.}
\end{cases}
$$ 
\end{defn}

\begin{defn}\cite{Bksrs}
Let $F$ be a function from $\F_{2^n}$ to itself. For a fixed $x_0 \in \F_{2^n}$, we call $F(x)$ as $x_0$-APN(or partially APN) if all the points $x, y$ satisfying
$$F(x_0)+F(x)+F(y)+F(x_0 +x+y) = 0 $$
belong to the curve $(x_0 +x)(x_0 +y)(x+y) = 0.$
\end{defn}

\begin{defn}~\textup{\cite{Bouk,LYT}} For $F : \F_{q} \rightarrow \F_{q}$, $q=p^n$, $n$ a positive integer, $p$ an arbitrary prime,   and $a, b \in \F_{q}$, the {\em second-order zero differential spectra} of $F$ with respect to $a,b$ is defined as
\begin{equation}
\nabla_F(a, b) := \# \{X \in \F_{q}: F(X + a + b) -F (X + b) -F(X + a) + F (X) = 0\}.
\end{equation}
\end{defn}
If $q = 2^n$, we call $\nabla_F = \text{max}\{\nabla_F(a,b):a\neq b, a, b \in \F_{2^n}^* \}$ the {\em second-order zero differential uniformity} of $F$. If $q=p^n,\,p > 2$, we call $\nabla_F =\text{max}\{\nabla_F(a,b) : a, b \in \F_{p^n}^* \}$ the {\em second-order zero differential uniformity} of $F$.

\begin{defn}{(Feistel Boomerang Connectivity Table)}~\textup{\cite{Bouk}}
 Let $F:\F_{2^n}\to \F_{2^n}$ and $a,b \in \F_{2^n}$. The {\em Feistel Boomerang Connectivity Table (FBCT)} of $F$ is given by a $2^n \times 2^n$ table, in which the entry for the $(a, b)$ position is given by
 \begin{equation*}
FBCT_F(a, b) =\# \{X \in \F_{2^n}: F(X + a + b) +F (X + b) +F(X + a) + F (X) = 0\}.
\end{equation*}
\end{defn}

\begin{defn}{($F$-Boomerang Uniformity)}~\textup{\cite{Bouk, LYT}}
The $F$-Boomerang uniformity corresponds to the highest value in the FBCT without considering the first row, the first column and the diagonal, that is,
$$
\beta_F = \max_{a \neq 0, b \neq 0, a \neq b} FBCT_F(a,b).
$$
\end{defn}
Note that the coefficients of FBCT are exactly the second-order zero differential spectra of functions over $\F_{2^n}$, and that the $F$-Boomerang uniformity is in fact the second-order zero differential uniformity of $F$ in even characteristic.

The following  lemma gives a characterization of APN functions in terms of their second-order zero differential uniformity.
\begin{lem}\label{L00}\cite[Lemma 2.6]{LYT}
F is an APN function of $\F_{2^n}$ if and only if
\begin{equation*}
\nabla_F(a,b)=
\begin{cases}
0 &~\mbox{~if~} ab \neq0, a\neq b,\\
2^n &~\mbox{~otherwise}.
\end{cases}
\end{equation*}
where $a,b \in \F_{2^n}$. Moreover, $F$ is second-order zero differential $0$-uniform.
\end{lem}
The characterization of PN functions is given by the following lemma.
\begin{lem}\label{L0}\cite[Lemma 2.5]{LYT}
F is a PN function of $\F_{p^n}$ if and only if
\begin{equation*}\nabla_F(a,b)=
\begin{cases}
0 &~\mbox{~if~} ab \neq0,\\
p^n &~\mbox{~otherwise}.
\end{cases}
\end{equation*}
where $a,b \in \F_{p^n}$. Moreover, $F$ is second-order zero differential $0$-uniform.
\end{lem}

We will be using the following result by Coulter and Henderson~\cite{RH} to determine the roots of trinomial $X^{p^k}-AX-B$ in some of the results.

\begin{lem}
\label{L1} 
\cite[Theorem 3]{RH}
 Let $k$ be a non-negative integer and $F(X) = X^{p^k}-AX-B \in \F_{p^n}[X]$, $A \neq 0$. Let $d= \gcd(k,n)$, $m=n/d$ and $\Tr_d^n$ be the relative trace from $\F_{p^n}$ to $\F_{p^d}$. For $0\leq i \leq m-1$, define $\displaystyle t_i=\sum_{j=i}^{m-2} p^{k(j+1)}$. Put $\alpha_0 = A$ and $\beta_0=B$. If $m>1$, then for $1\leq r \leq m-1$, we let $\alpha_r = A^{1+p^k+p^{2k}+\cdots+p^{kr}}$ and $\                                      \displaystyle \beta_r = \sum_{i=0}^{r}A^{s_i} B^{p^{ki}}$ where $\displaystyle s_i = \sum_{j=i}^{r-1} p^{k(j+1)}$ for $0 \leq i \leq r-1$ and $s_r =0$.
 \begin{enumerate}[(i)]
\item If $\alpha_{m-1} = 1$ and $\beta_{m-1} \neq 0$ then the trinomial $F$ has no roots in $\F_{p^n}$.
\item If $\alpha_{m-1}\neq 1$ then $F$ has a unique root, namely $X = \dfrac{\beta_{m-1}}{1-\alpha_{m-1}}$.
\item If $\alpha_{m-1} = 1, \beta_{m-1} = 0$, $F$ has $p^d$ roots in $\F_{p^n}$  given by $x+\delta\tau$, where $\delta \in \F_{p^d}$, $\tau$ is fixed in $\F_{p^n}$ with $\tau^{p^k-1} = A$ (that is, a $(p^k-1)$-root of $A$), and, for any $c \in \F_{p^n}^{*}$, satisfying $\Tr_d (c)\neq 0$ then $\displaystyle x = \frac{1}{\Tr_d (c)} \sum_{i=0}^{m-1} \left(\sum_{j=0}^{i} c^{p^{kj}}\right) A^{t_i} B^{p^{ki}}$. 
\end{enumerate}
\end{lem}

The following two lemmas discuss the solutions of quadratic and cubic equations, respectively.
\begin{lem} \label{L3}\cite[Lemma 11]{EFRS}
 Let $n$ be a positive integer. The equation $X^2+aX+b=0$, with $a,b \in \F_{2^n}$, $a \neq 0$, has two solutions in $\F_{2^n}$ if $\Tr\left(\frac{b}{a^2}\right)=0$, and no solution otherwise.
\end{lem}

\begin{lem}\label{L4}\cite[Theorem 1]{William}
The factorisation of $F(X)=X^3+X+a$ over $\F_{2^n}$ with $a \neq 0$ is characterized as follows:
\begin{enumerate}
 \item $F = (1, 1, 1)$ (a product of three linear factors) if and only if $\Tr\left(\frac{1}{a}\right)=\Tr(1)$, $t_1$ and $t_2$ are cubes in $\F_{2^n}$(n even), $\F_{2^{2n}}$(n odd), where $t_1$, $t_2$ are roots of $t^2 + at + 1 = 0$.
 \item $F = (1, 2)$  (a product of one linear and one quadratic factors)  if and only if $\Tr\left(\frac{1}{a}\right)\neq \Tr(1)$.
 \item $F = (3)$ (irreducible) if and only if $\Tr\left(\frac{1}{a}\right)=\Tr(1)$, $t_1$ and $t_2$ are not cubes in $\F_{2^n}$(n even), $\F_{2^{2n}}$(n odd), where $t_1$, $t_2$ are roots of $t^2 + at + 1 = 0$.
\end{enumerate}
\end{lem}

Recently,  Budhagyan and  Pal~\cite{BP} obtained a new class of binomials with low differential uniformity, stated below.  In the next section, we will look at its second-order zero differential uniformity.

\begin{lem}\cite[Theorem 4.2]{BP}\label{L2}
Let $p$ be an odd prime and $n$ be a positive integer. Then the binomial $F(X) = X^{p^n-1}+uX^2$, $u \in \F_{p^n}^{*}$ over $\F_{p^n}$ is APN if
\begin{equation*}
\begin{cases}
\chi(u)=-1 &~\mbox{~and~} p^n \equiv 1(\mbox{~mod~} 4)\\
\chi(u)=1 &~\mbox{~and~} p^n \equiv 3(\mbox{~mod~} 4).
\end{cases}
\end{equation*}
and differentially $3$-uniform, otherwise.
\end{lem}
\section{The second-order zero differential spectrum for functions over finite fields of odd characteristic}\label{S3}

Boukerrou et al.~\cite{Bouk} showed that all the non trivial
second-order zero differential spectra of APN functions in even characteristic are $0$ (see Lemma~\ref{L00}). Further, Li et al.~\cite{LYT} gave the characterization of PN functions in terms of the second-order zero differential spectra over $\F_{p^n}$ for an odd prime $p$, as given in Lemma~\ref{L0}. Surely, a natural question is whether there are other constraints on the second-order zero differential spectra for APN functions in odd characteristic. In the following theorem, we show that for any odd or even function $F$ over $\F_{p^n}$, where $p$ is an odd prime, if the second-order differential uniformity is $1$, then $F$ must be an APN function.

\begin{thm}\label{T0}
Let $F(X)$ be an odd or even function over $\F_{p^n}$, where $p$ is an odd prime. If the differential uniformity $\Delta_F$ of $F$ is strictly greater than $2$, then the second-order zero differential uniformity $\nabla_F$ of $F$ is never equal to $1$. In other words, if  $\nabla_F=1$, then $F$ is an APN function over $\F_{p^n}$. 
\end{thm}
\begin{proof}
It is given that the function $F(X)$ has differential uniformity strictly greater than $2$. Assume on the contrary that $\nabla_F =1$.  Now, for $a,b \in \F_{p^n}\setminus\{0\}$, we consider the equation
\begin{equation}\label{eq}
F (X + a + b) -F(X + b) -F(X + a)+ F(X) = 0.
\end{equation}
Since $F$ is odd or even, it is straightforward to see that if $X$ is a solution of Equation~\eqref{eq} then $-(X+a+b)$  is also a solution of  Equation~\eqref{eq}. Since $\Delta_F \geq 3$, there exists $(a',b') \in \F_{p^n}  \setminus \{0\} \times \F_{p^n} $ such that $\Delta_F(a',b') \geq 3$, that is, the difference equation $F(Z+a')-F(Z)=b'$ has at least three distinct solutions in $\F_{p^n}$. Suppose $X_1,X_2$ and $X_3$ are any three distinct solutions to $F(Z+a')-F(Z)=b'$ in $\F_{p^n}$. We may write $X_2=X_1+u$ and $X_3=X_2+v$ for some $u,v \in \F_{p^n}^{*}$. We now compute $\nabla_F(u,a')$ by simply looking at the solutions to the equation
\begin{equation}\label{eq1}
F(X + u + a') -F(X + u) -F(X + a') + F(X) = 0.
\end{equation}
It follows immediately that $X=X_1$ is a solution to the above Equation ~\eqref{eq1}, and as a consequence, $\nabla_F(u,a') \geq 1$. However, by our assumption that the second-order zero differential uniformity of $F$ is 1, we must have  $\nabla_F(u,a')=1$. As per our earlier observation, $-(X_1+u+a')$ is also a solution to Equation ~\eqref{eq1}. Thus, in view of the fact that $\nabla_F(u,a')=1$, we must have that $X_1=-(X_1+u+a')$, i.e., $X_1=\dfrac{-(u+a')}{2}$. Hence, $X_2=X_1+u=\dfrac{u-a'}{2}$. Similarly, one can show that $\nabla_F(v,a')=1$, with $X_2$ satisfying the equation
\begin{equation}\label{eq2}
 F(X + v + a') -F(X + v) -F(X + a') + F(X) = 0.
\end{equation}
This will yield us $X_2=\dfrac{-(v+a')}{2}$ by following the similar arguments as before. Now, comparing  both expressions for $X_2$, we obtain $u=-v$, which implies that $X_1=X_3$. Thus, we get a contradiction to the fact that $X_1$ and $X_3$ are distinct. Together with Lemma 2.5, this gives that, if  $\nabla_F=1$, then $F$ is an APN function over $\F_{p^n}$.  This completes the proof of both statements of the theorem.
\end{proof}

\begin{rmk}
 The converse of Theorem~\textup{\ref{T0}} does not hold. For instance, consider the function $F(X)=X^{14}$, which is APN over the finite field $\F_{5^3}$ but has second-order zero differential uniformity $4$.
 Moreover, the statement of Theorem~\textup{\ref{T0}} does not hold in general. To illustrate this, consider the function $F(X)=X^9+X^{4}$, which is a differentially $3$-uniform function over $\F_{11}$. However, it has second-order zero differential uniformity $1$.
\end{rmk}

\begin{rem} 
\label{r2} 
Note that the proof implies that, for $F$ even or odd ($p$ odd), if $\nabla_F=1$, if for a certain $a$ and $b$ there is a solution $X$, then the solution is $X=-\frac{a+b}{2}$, since $X=-X-a-b$.
Further, we note that, if $F$ is odd, then $X=-\frac{a+b}{2}$ is always a solution.
\end{rem}

\begin{cor} If $F$ is an odd function, then $F$ is not PN.
\end{cor}
\begin{proof} It follows directly from Remark \ref{r2} and Lemma \ref{L0}.
\end{proof}
\begin{rem}
\label{r3}
Note that the proof also implies that, if  $F$ is an even function, the number of solutions to the equation is always even, unless $X=-X-a-1$ for some solution $X$, in which case $X=\frac{-a-1}{2}$. This implies that, unless $X=\frac{-a-1}{2}$ is a solution, the second-order zero differential uniformity of any even function $F$ is even.
\end{rem}

\begin{rem} 
It is worth noting that since both $\Delta_F$ and $\nabla_F$ are invariant under EA (extended affine) equivalence, the theorem is also true for a function EA-equivalent to an even or odd function. Recall that two vectorial functions $F, G$ on $\F_{p^n}$ (for any characteristic $p$) are EA-equivalent if  that there exist two affine permutations $P,Q$ and an affine function $A$ satisfying $G = P\circ F\circ Q+A$.  
\end{rem} 

We next compute the second-order zero differential spectrum of the differentially low uniform binomial  given in Lemma~\ref{L2} (which is APN under some conditions of the parameter $u$). 

\begin{thm}
Let $p$ be an odd prime and $n$ be a positive integer. Then the binomial $F(X) = X^{p^n-1}+uX^2$, $u \in \F_{p^n}^{*}$ over $\F_{p^n}$. Then for $a,b \in \F_{p^n}$,
\begin{equation*}
\nabla_F(a,b)=
\begin{cases}
p^n &~\mbox{~if~} ab=0,\\
1 &~\mbox{~if~} ab \neq0, p\neq3, a=-b \mbox{~and~} -1+abu=0,\\
2 &~\mbox{~if~} ab \neq0, p\neq3, a=b \mbox{~and~} -1+2abu=0 \mbox{~or~} 1+abu=0,\\
 &~\mbox{~or~} ab \neq 0, p\neq3, a=- b \mbox{~and~} 1+2abu=0, \\
&~\mbox{~or~} ab \neq 0, a\neq b,-b \mbox{~and~} -1+2abu=0  \mbox{~or~}1+2abu=0,\\
3 &~\mbox{~if~} ab \neq0, p=3, a=-b  \mbox{~and~}  -1+abu=0,\\
4&~\mbox{~if~} ab \neq0, p=3, a=b \mbox{~and~} 1+abu=0,\\
0&~\mbox{~otherwise}.
\end{cases}
\end{equation*} 
Notice that $\nabla_F(a,b)\geq3$ is possible only when $p=3$. Hence, $F$ is second-order zero differentially $2$-uniform if $p \neq 3$ and when $p=3$ it is second-order zero differential $4$-uniform.  
\end{thm}
\begin{proof}
 Now for $a,b \in \F_{p^n}$, we consider the equation
\begin{equation*}
F(X + a + b) -F (X + b) -F(X + a) + F (X) = 0.
\end{equation*}
If $ab = 0$, then $\nabla_F(a, b) = p^n$.
If $ab \neq 0$, then we can write the above equation as
\begin{equation}\label{e2}
(X + a + b)^{p^n-1} - (X + b)^{p^n-1} -(X + a)^{p^n-1} + X^{p^n-1} + 2abu= 0.
\end{equation}
Now we divide our analysis in the following five cases.
\noindent \\
\textbf{Case 1.} Let $X=0$ satisfies Equation~\eqref{e2}, then 
$$
(a + b)^{p^n-1} - b^{p^n-1} -a^{p^n-1} + 2abu= 0.
$$
For $(a,b) \in \F_{p^n} \times \F_{p^n}$, when $a=-b$, $X=0$ is a solution if $-1+abu=0$. When $a \neq -b$, $X=0$ is a solution if $-1+2abu=0.$
\noindent \\
\textbf{Case 2.} Let $X=-a$ satisfies Equation~\eqref{e2}, then 
$$
b^{p^n-1} - (-a+b)^{p^n-1} +a^{p^n-1} + 2abu= 0.
$$
For $(a,b) \in \F_{p^n} \times \F_{p^n}$, when $a=b$, $X=-a$ is a solution if $1+abu=0$. When $a \neq b$, $X=-a$ is a solution if $1+2abu=0.$
\noindent \\
\textbf{Case 3.} Let $X=-b$. Due to the symmetry in Case 2 and Case 3, we have $X=-b$ is a solution of Equation~\eqref{e2} if $a=b$ and $1+abu=0$ or $a \neq b$ and $1+2abu=0.$
\noindent \\
\textbf{Case 4.} Let $X=-(a+b)\neq0$. Then Equation~\eqref{e2} has a solution only if $-1+2abu=0.$ Note that this excludes $a=-b$, in which case we obtain Case 1.
\noindent \\
\textbf{Case 5.} Let $X\not \in \{0,-a,-b,-(a+b)\}$. Then Equation~\eqref{e2} has no solution in $\F_{p^n}$.
Summarizing the above cases, we can get the desired claim.
\end{proof}

Dobbertin et al.~\cite{DHKM} showed that the power function  $F(X)=X^{2\cdot3^{\frac{n-1}{2}}+1}$ over $\F_{3^n}$ for odd $n$ is differentially  $4$-uniform. We next compute its second-order zero differential spectrum.
\begin{thm}
Let $F(X)=X^d$, where $d=2\cdot3^{\frac{n-1}{2}}+1$, be a function over $\F_{3^n}$, for odd $n$. Then for $a,b \in \F_{3^n}$,
\begin{equation*}
\nabla_F(a,b)=
\begin{cases}
p^n &~\mbox{~if~} ab=0 \mbox{~and~} a=b,\\
3 &~\mbox{~if~} \left(\dfrac{-ab^{3^m}}{ab^{3^m}+a^{3^m}b}\right)^{1+3^m+3^{2m}+\cdots+3^{2m^2}} =1,\\
1 &~\mbox{~otherwise}.\\
\end{cases}
\end{equation*}
Moreover, $F$ is second-order zero differentially $3$-uniform.
\end{thm}
\begin{proof}
 For $a,b \in \F_{3^n}$, consider the following equation   
\begin{equation*}
F(X + a + b) -F (X + b) -F(X + a) + F (X) = 0.
\end{equation*}
If $ab = 0$, then $\nabla_F(a, b) = p^n$.
If $ab \neq 0$, then we can write the above equation as
$$
(X + a + b)^{d} -(X + b)^{d} -(X + a)^{d} + X^{d} = 0,
$$
or equivalently,
$$
(Y + B + 1)^{2\cdot3^{m}+1} -(Y + B)^{2\cdot3^{m}+1} -(Y + 1)^{2\cdot3^{m}+1} + Y^{2\cdot3^{m}+1} = 0,
$$
where $m=\dfrac{n-1}{2}, Y=\dfrac{X}{a}$ and $B=\dfrac{b}{a}$. The above equation can be further reduced to
$$
 (B+B^{3^m})Y^{3^m}+B^{3^m}Y+B^{3^m+1}+B^{3^m}-B^{2\cdot3^m}-B=0,
$$
which is equivalent to
\begin{equation}~\label{e3q}
 (B+B^{3^m})Y^{3^m}+B^{3^m}Y-(B^{3^m}-B)(B^{3^m}-1)=0.
\end{equation}
If $B=\pm 1$ (i.e. $a=\pm b$), then the above Equation~\eqref{e3q} has three solutions, namely, $0,1$ and $2$. Let us now assume $a \neq \pm b$ and rewrite Equation~\eqref{e3q} as follows
\begin{equation}~\label{e31q}
 Y^{3^m}+\left(\frac{B^{3^m}}{B+B^{3^m}}\right)Y-\frac{(B^{3^m}-B)(B^{3^m}-1)}{B+B^{3^m}}=0.
\end{equation}
We know from Remark ~\ref{r2} that for every entry $(1,B) \in  \F_{p^n} \times \F_{p^n}$, there exists at least one solution $Y=-\dfrac{1+B}{2}$, as $d$ is odd. Hence, using Lemma~\ref{L1}, we can further characterize the solutions of Equation~\eqref{e31q}. In fact, if $\alpha_{n-1}=\left(\dfrac{-B^{3^m}}{B^{3^m}+B}\right)^{1+3^m+\cdots+3^{(n-1)m}}=1$, Equation~\eqref{e31q} has three solutions in $\F_{p^n}$, and otherwise one solution in $\F_{p^n}$.
\end{proof}

In the following theorem, we consider the second-order zero differential uniformity of an APN function introduced in~\cite{HRS99}.
\begin{thm}
Let $p$ be an odd prime and $n$ be a positive integer. Let the APN function $F(X) = X^d$ over $\F_{p^n}$, where $d=\frac{p^n+1}{4}+\frac{p^n-1}{2}$, if $p^n\equiv 3\pmod 8$, or $d=\frac{p^n+1}{4}$, if $p^n\equiv 7\pmod 8$. Then the second-order zero differential uniformity of $F$ is given by
\begin{itemize}
\item[$(1)$] if $p=3$, then $\nabla_F\leq 8$;
\item[$(2)$]  if $p>3$, then $\nabla_F\leq8$, if $\chi(2)=-1$, and $\nabla_F\leq18$, if $\chi(2)=1$.
\end{itemize}
\end{thm} 
Note that, under the conditions of our theorem, that is, $p^n\equiv 3,7\pmod 8$,  $\chi(2)=1$ if and only if $p^n\equiv 7\pmod 8$. 
\begin{proof}
The second-order zero differential uniformity equation for $F$ at $(a,b)$ (where $ab\neq 0$, or $a\neq -b$, since those cases are trivial) is
\[
(X+a+b)^d-(X+b)^d-(X+a)^d+X^d=0,
\]
which,  dividing by $b^d$ and using the relabelling $X/b\mapsto X, a/b\mapsto a$, is equivalent to
\begin{equation}
\label{eq:x}
(X+a+1)^d-(X+1)^d-(X+a)^d+X^d=0.
\end{equation}
If $X=0,-1,-a,-a-1$, then $a$ must satisfy (since $d$ is even) $(a+1)^d-a^d=1$, $(a-1)^d-a^d=1$, $(a-1)^d-a^d=1$, $(a+1)^d-a^d=1$, respectively.  All of these cases allow only two values of $a$, since $F$ is APN.  In fact, from the proof of the APN-ness of this function in \cite{HRS99}, there are no other solutions than $a=0$, which is excluded.
We thus cannot get the solutions $X=0,-1,-a,-a-1$. 

We now assume that $X\not\in\{0,-1,-a,-a-1\}$.
We next divide by $X^d$ and label $Y=1/X$. We thus get
\[
(1+(a+1)Y)^d-(1+Y)^d-(1+aY)^d+1=0.
\]
We set $U_{\alpha Y}=(1+\alpha Y)^d$, where $\alpha=1,a,a+1$. Observe that, as in \cite{HRS99}, $2d\equiv \frac{p^n+1}{2}\pmod {p^n-1}$, and so, $U_{\alpha Y}^2=\chi(1+\alpha Y)\, (1+\alpha Y)$. Squaring the second-order zero differential uniformity equation  $U_Y+U_{aY}=U_{(a+1)Y}+1$, we get
\begin{equation}
\label{eq:maj3}
\chi(1+Y) (1+Y)+\chi(1+aY) (1+a Y)+2 U_Y U_{aY}=\chi(1+(a+1)Y) (1+(a+1)Y) +1+2 U_{(a+1)Y}.
\end{equation}
We now consider $8$  (though we combine some of these) cases depending upon the sign vector 
\[
\left(\chi(1+Y), \chi(1+aY), \chi(1+(a+1)Y)\right)=(\alpha,\beta,\gamma)\in\{\pm 1\}^3.
\]

\noindent
{\bf Case 1.} Let $(\alpha,\beta,\gamma)=(1,1,1)$.
Equation~\eqref{eq:maj3} becomes 
\begin{align*}
&1+Y+1+aY+2 U_Y U_{aY}=1+(a+1)Y+1+2 U_{(a+1)Y}, \text {that is},\\
& U_Y U_{aY}=  U_{(a+1)Y}, \text{ which, by squaring again, renders}\\
& (1+Y)(1+aY)=1+(a+1)Y, \text{ implying (since $Y\neq 0$)}\\
&a=0, \text{ an impossibility}.
\end{align*}

\noindent
{\bf Case 2.} Let $(\alpha,\beta,\gamma)=(1,1,-1)$.
Equation~\eqref{eq:maj3} becomes 
\begin{align*}
&1+Y+1+aY+2 U_Y U_{aY}=-(1+(a+1)Y)+1+2 U_{(a+1)Y}, \text {that is},\\
&  U_Y U_{aY}=  -1-(a+1)Y +U_{(a+1)Y}, \text{ which, by squaring as before, renders}\\
& (1+Y)(1+aY)=(1+(a+1)Y)^2-(1+(a+1)Y)-2(1+(a+1)Y)U_{(a+1)Y},\\
& (a^2+a+1) Y^2-1=2(1+(a+1)Y)U_{(a+1)Y}, \text{ and squaring again}\\
& (a^2+a+1)^2 Y^4-2 (a^2+a+1) Y^2+1=-4 (1+(a+1)Y)^2 (1+(a+1)Y),\\
&5 + 12 (1 + a) Y + 2 (5 + 11 a + 5 a^2) Y^2 +  4 (1 + a)^3 Y^3 + (1 + a + a^2)^2 Y^4=0,
\end{align*}
which has at most four solutions. The discriminant of this equation is $-2^8 (-1 + a)^4 a^3 (8 + 11 a + 8 a^2)$. Recall that for any polynomial of any degree $m$, a necessary condition for the polynomial to be irreducible over $\F_{p^n}$ is to have its discriminant a square ($m$ odd), respectively, non-square ($m$ even) (see~\cite{Dickson06}). Thus, we have  no solutions if $-a(8 +11 a + 8 a^2)$ is a non-square.

\noindent
{\bf Case 3.} Let $(\alpha,\beta,\gamma)=(1,-1,1)$.
Equation~\eqref{eq:maj3} becomes 
\begin{align*}
&1+Y-(1+aY)+2 U_Y U_{aY}=1+(a+1)Y+1+2 U_{(a+1)Y},\\
&U_Y U_{aY}=U_{(a+1)Y}+1+aY,\text{ and by squaring}, \\
&-(1+Y)(1+aY)=1+(a+1)Y+1+a^2Y^2+2aY+2(1+aY)U_{(a+1)Y},\\
&-3 -2 (1 + 2a) Y - (a+a^2) Y^2= 2(1+aY)U_{(a+1)Y},\\
&5 + 4 (2 + 3 a) Y + 2 (1 + a) (2 + 5 a) Y^2 + 4 a (1 + a)^2 Y^3 +  a^2 (1 + a)^2 Y^4=0,
\end{align*}
which again has at most four solutions. The discriminant of this equation is $2^8 a^3 (1 + a)^4 (8 - 11 a + 8 a^2)$, so we have  no solutions if $a(8 - 11 a + 8 a^2)$ is a non-square.

\noindent
{\bf Case 4.} Let $(\alpha,\beta,\gamma)=(1,-1,-1)$.
In this case, we have the sequence of implications
\begin{align*}
&1+Y-(1+aY)+2 U_Y U_{aY}=-(1+(a+1)Y)+1+2 U_{(a+1)Y},\\
& U_Y U_{aY}=U_{(a+1)Y}-Y,\\
&-(1+Y)(1+a Y)=-(1+(a+1)Y)+Y^2-2Y U_{(a+1)Y}\\
&(1 + a) Y^2=2Y U_{(a+1)Y}, \text{ and  dividing by $Y$ and squaring we get}\\
&(1 + a)^2 Y^2=-4 (1+(a+1)Y),\\
& (2 + (1+a)Y)^2=0, \text{ of solution } Y=\frac{-2}{a+1}.
\end{align*}

\noindent
{\bf Case 5.} Let $(\alpha,\beta,\gamma)=(-1,1,1)$.
As before, we get
\allowdisplaybreaks
\begin{align*}
&-(1+Y)+(1+aY)+2 U_Y U_{aY}=(1+(a+1)Y)+1+2 U_{(a+1)Y},\\
&U_Y U_{aY}=U_{(a+1)Y}+Y+1,\\
&-(1+Y)(1+a Y)=(1+(a+1) Y)+(Y+1)^2+2 (Y+1) U_{(a+1)Y},\\
& -3 - 2 (2 + a) Y - (1 + a) Y^2=2 (Y+1) U_{(a+1)Y},\text{ and so},\\
& 5 + 4 (3 + 2 a) Y + 2 (1 + a) (5 + 2 a) Y^2 +  4 (1 + a)^2 Y^3 + (1 + a)^2 Y^4=0,
\end{align*}
with at most four solutions. The discriminant of this equation is $2^8 a^3 (1 + a)^4 (8 - 11 a + 8 a^2)$, so we have  no solutions if $a(8 - 11 a + 8 a^2)$ is a non-square. 

\noindent
{\bf Case 6.} Let $(\alpha,\beta,\gamma)=(-1,1,-1)$.
We have
\allowdisplaybreaks
\begin{align*}
&-(1+Y)+1+a Y+2 U_Y U_{aY}=-(1+(a+1) Y)+1 +2 U_{(a+1)Y},\\
 &U_Y U_{aY}=U_{(a+1)Y}- a Y,\\
 &-(1+Y)(1+a Y)=-(1+(a+1) Y)+a^2 Y^2-2 a Y U_{(a+1)Y}\\
 &2 a Y U_{(a+1)Y} = a (1 + a) Y^2, \text{ and dividing by $aY$ and squaring},\\
 & -4(1+(a+1)Y)=(1+a)^2 Y^2, \text{ that is},\\
 & (2 + Y + a Y)^2, \text{ of solution } Y=\frac{-2}{a+1}.
\end{align*}

\noindent
{\bf Case 7.} Let $(\alpha,\beta,\gamma)=(-1,-1,1)$.
As before, we look at the equation
\allowdisplaybreaks
\begin{align*}
&-(1+Y)-(1+a Y)+2 U_Y U_{aY}=(1+(a+1) Y)+1 +2 U_{(a+1)Y},\\
&   U_Y U_{aY}=2+  (1 + a) Y+U_{(a+1)Y},\\
& (1+Y)(1+a Y)= (2+  (1 + a) Y)^2+(1+(a+1)Y)+2(2+  (1 + a) Y)U_{(a+1)Y},\\
& -4 - 4 (1 + a) Y + (-1 - a - a^2) Y^2=2(2+  (1 + a) Y)U_{(a+1)Y},\\
& (-4 - 4 (1 + a) Y + (-1 - a - a^2) Y^2)^2=4(2+  (1 + a) Y)^2(1+(a+1)Y),\text{ and so},\\
&4(1 +  a^2) + 4(1 +  a)(1+a^2)   Y + (1 + a + a^2)^2 Y^2=0,
\end{align*}
which has at most two solutions.

The discriminant of this equation is $-2^4 a^2 (1 + a^2)$. Thus, we have no solutions if $\chi(-(a^2+1))=-\chi(a^2+1)=-1$, that is, $a^2+1$ is a square and two solutions, otherwise.

\noindent
{\bf Case 8.} Let $(\alpha,\beta,\gamma)=(-1,-1,-1)$.
We finally have
\allowdisplaybreaks
\begin{align*}
&-(1+Y)-(1+a Y)+2 U_Y U_{aY}=-(1+(a+1) Y)+1 +2 U_{(a+1)Y},\\
&U_Y U_{aY}=U_{(a+1)Y}+1=U_Y +U_{aY}, \text{ and dividing by $U_Y$}\\
&U_{aY}-1= \left(\frac{1+aY}{1+Y} \right)^d, \text{ and by squaring}\\
&-(1+aY) +1-2 U_{aY}=\frac{1+aY}{1+Y}, \text{ that is},\\
& -2 U_{aY}= \frac{a Y^2+2 a Y+1}{Y+1},\\
&5 + 8 (1 + a) Y + 2 (2 + a) (1 + 2 a) Y^2 + 4 a (1 + a) Y^3 + a^2 Y^4=0,
\end{align*}
 which has at most four solutions. The  discriminant of this equation is $ -2^8 (-1 + a)^4 a^3 (8 + 11 a + 8 a^2)$, so we have no solutions if and only if $-a(8 + 11 a + 8 a^2)$ is a non-square.
 
 Recall that $\chi(-1)=-1$ under our conditions on $n,p$. 
 Surely, Cases 4 and 6  cannot occur simultaneously, and we can see that the conditions on $Y$ in Case 4 (as well, as Case 6) are fulfilled for $a\neq\pm1$, so one of these cases can happen, but not both.

 Note also that the conditions for Cases 2 and 8 are identical, as are the conditions for Cases 3 and 5. 

In Cases 4 and 6, the solution is $ Y=\frac{-2}{a+1}$. Then:
\begin{align*}
&\chi(1+(a+1)Y)=\chi(-1)=-1.\\
&\chi(1+Y)=\chi\left(\frac{a-1}{a+1}\right)=\chi(a-1)\chi(a+1).\\
&\chi(1+aY)=\chi\left(\frac{-a+1}{a+1}\right)=\chi(-1)\chi(a-1)\chi(a+1)=-\chi(a-1)\chi(a+1)=-\chi(1+Y).\\
\end{align*}
For any $a\neq\pm1,0$, then, we have that either $\chi(a-1)=\chi(a+1)$ (Case 4) or $\chi(a-1)=-\chi(a+1)$ (Case 6), so $ Y=\frac{-2}{a+1}$ is always a solution 
 of Equation (\ref{eq:maj3}). However, by replacing $X=\frac{1}{Y}=-\frac{a+1}{2}$ in the original equation $(X+a+1)^d-(X+1)^d-(X+a)^d+X^d=0$,
 we can see that the latter equation is not satisfied by this value of $X$:
\begin{align*}
&\left(-\frac{a+1}{2}+a+1\right)^d-\left(-\frac{a+1}{2}+1\right)^d-\left(-\frac{a+1}{2}+a\right)^d+\left(-\frac{a+1}{2}\right)^d=0,\\
&2^{-d+1}((a+1)^d-(a-1)^d)=0,\text{ that is},\\
&\left(\frac{a-1}{a+1}\right)^d=1.
\end{align*}
The number of solutions to the equation $z^d=1$  is given by $\gcd(d,p^n-1)=2$. It is easy to see that these solutions are $z=\pm1$. However, $z=1$ renders a contradiction, while $z=-1$ renders  $a=0$, which is excluded. This implies that, for $a\neq\pm1$,  Cases 4 and 6 do not contribute to the total number of solutions.
It is easy to see that $a=-1$ renders a contradiction on Cases 4 and 6, while, trivially, the solution $a=\pm1$ cannot fulfill the conditions on these cases.

Therefore, for any $a\in\F_{p^n}^*$, Cases 4 and 6 do not contribute to the total number of solutions of the original equation (\ref{eq:x}). 

 We can prune the other cases a bit better in the following way. We express the polynomials in the ``variable'' $a$, and compute their discriminant. If $a$ is fixed in $\F_{p^n}$ and $Y_i$ exists for the Case~$i$ ($i=2,3,5,7,8$), then the corresponding discriminant of the  polynomial in $a$ cannot be a square. These discriminants for Cases $2,3,5,7,8$ are:
 \[
 2 Y_2^{15} (1 + Y_2)^6, Y_3^{15} (1 + Y_3)^3, Y_5^4 (1 + Y_5)^3 (2 + Y_5)^2,  2 Y_7^2 (1 + Y_7)^3, 2 Y_8^2 (1 + Y_8)^3.
 \]
 Since $\chi(Y_i+1)=1$ in Cases 2 and 3, and  $\chi(Y_i+1)=-1$ in the other cases, we see that, if $\chi(2)=-1$, we can exclude Cases 7 and 8, since the discriminant is therefore a square, and, since the discriminant (with variable $X$) for Case 2 is the same as for Case 8, we can also exclude this case; we thus obtain if $\chi(2)=-1$ an upper bound of 8.  If $\chi(2)=1$, then we obtain an upper bound of 18. 

 If $p=3$, then we can further refine the  distribution of the potential solutions. 
In this case,    since $\chi(2)=-1$, we can exclude Cases 2, 7 and 8. Next, $8a^2-11a+8=2(a^2-a+1)=2(1 + a)^2$, which is a square only if $a=-1$.
 
 First, if $p=3$, $a= -1$, then  
 in Case 3 the equation reduces to $2-Y=0$, which renders as the only solution $Y=2$.  In Case 5 the equation reduces to $2+Y=0$, which renders as the  only solution $Y=-2$. 
 Therefore, when $a=-1$ there are exactly two solutions.
 
 Next, let $p=3$, $a\not=- 1$.  
The maximum number of solutions can only happen for Cases 3 and 5 (when $\chi(a)=-1$).  
 When $\chi(a)=1$, there are no solutions.
\end{proof}
 \begin{rmk}
It is worth noting that, since $F$ is an even function and $X=\frac{-a-1}{2}$ is not a solution, by Remark~\textup{\ref{r3}} the  second-order zero differential uniformity of $F$ is even.\end{rmk} 

We display in Table~\ref{Table2} some experimental data on the monomial function of our previous theorem, which shows that our bound of $8$ (when $\chi(2)=-1$) for the second-order zero differential uniformity is attained for some values of $p$ and $n$.
\begin{table}[hbt]
\caption{Experimental results on $X^d$}
\label{Table2}
\begin{center}
\begin{tabular}{|c|c|c|c|} 
 \hline
 $p$ & $n$ & $d$ & $\nabla_F$\\
 \hline
 $3$ &$3$ &20 & $2$\\
 \hline
 $3$ &$5$ & 182&  $4$\\
 \hline
  $3$ &$7$ & 1640&  $8$\\
 \hline
  $3$ &$9$ &14762&$6$\\
 \hline
 $7$ &$3$ & 86&  $4$\\
 \hline
 $7$ &$5$ &4202&  $8$\\
 \hline
$11$ &$3$&998 &  $6$\\
 \hline
\end{tabular}
\end{center}
\end{table}

\section{The second-order zero differential spectrum of 0-APN functions and extensions}
\label{S4}

Budaghyan et al.~\cite{Bkrs} considered the function $F(X)=X^{21}$ and showed that it is $0$-APN if and only if $n$ is not a multiple of $6$ (that is, this function is an almost exceptional partial APN). As a side note, we  point out that A. Pott~\cite{Loops19} showed (non-constructively) that there exist partially APN permutations (in even characteristic), for all  integers~$n\geq 3$, by connecting the problem to Steiner triple systems. 
We discuss here the second-order zero differential spectrum of $F(X)=X^{21}$ and show that it attains the best possible value of FBCT for a non-APN function when $n$ is odd.

\begin{thm}\label{T1}
 Let $F(X)=X^{21}$ be a function over $\F_{2^n}$, where $n$ is odd. Then for $a,b \in \F_{2^n}$,
\begin{equation*}
\nabla_F(a,b)=
\begin{cases}
2^n &~\mbox{~if~} ab=0 \mbox{~or~} a=b,\\
4 &~\mbox{~if~} ab \neq0 \mbox{~and~} a\neq b, \displaystyle \sum_{i=0}^{n-1} v^{s_i}C^{2^{2i}}=0\\
&~\mbox{~and~} \displaystyle \sum_{r=0}^{n-1}\left(\sum_{i=0}^{n-1} (i+1) v^{s_i} C^{2^{2i}}\right)^{2^{2r}}=0,\\
0 &~\mbox{~otherwise,~}\\
\end{cases}
\end{equation*}
where $v=\dfrac{b}{a}\left(1+\dfrac{b^3}{a^3}\right)$, $C=1+\dfrac{b^{20}+a^4b^{16}+a^{15}b^5+a^3b^{17}}{a^{16}b^4+a^{19}b}$, $\displaystyle s_{i}= \sum_{j=i}^{n-2} 2^{2(j+1)}$ for $0\leq i\leq n-2$ and $s_{n-1}=0$. Moreover, the Feistel boomerang uniformity of $F$ is~$4$.

\end{thm}
\begin{proof}
 We consider the following equation for $a,b \in \F_{2^n}$,   
\begin{equation*}
F(X + a + b) +F (X + b) +F(X + a) + F (X) = 0.
\end{equation*}
If $ab = 0$ and $a=b$, then $\nabla_F(a, b) = 2^n$.
If $ab \neq 0$ and $a \neq b$, then we can write the above equation as
$$
(X + a + b)^{21} +(X + b)^{21} +(X + a)^{21} + X^{21} = 0,
$$
or equivalently,
$$
(Y + B + 1)^{21} +(Y + B)^{21} +(Y + 1)^{21} + Y^{21} = 0,
$$
where $Y=\dfrac{X}{a}$ and $B=\dfrac{b}{a}$ . The above equation can be further reduced to
\begin{align*}
 & (B^4+B)Y^{16}+(B^{16}+B)Y^{4}+(B^4+B^{16})Y+B^4+B^{17}+B^5+B^{16}+B^{20}+B=0.
\end{align*}
Notice that $B^4+B \neq 0$, because else we would have $a^3=b^3$ and as $n$ is odd, we get $\gcd(3,2^n-1)=1$ which would further imply $a=b$. Hence, dividing the above equation by $B^4+B$, we have 
\begin{equation}\label{ee}
  Y^{16}+\left(\frac{B^{16}+B}{B^4+B}\right)Y^{4}+\left(\frac{B^{16}+B^4}{B^4+B}\right)Y + \frac{B^4+B^{17}+B^5+B^{16}+B^{20}+B}{B^4+B}=0.
\end{equation}
One can further reduce the coefficients in the above equation as
\begin{align*}
 & \frac{B^{16}+B}{B^4+B} = 1+B^{12}+B^3+B^9+B^6 =:u , \\
 & \frac{B^{16}+B^4}{B^4+B} =B^{12}+B^3+B^9+B^6=u+1=:v.
\end{align*}
Thus, we can rewrite Equation~\eqref{ee} as $Y^{16}+uY^4+vY+ C=0,$ where
$$C= \dfrac{B^4+B^{17}+B^5+B^{16}+B^{20}+B}{B^4+B}.$$ 
Substituting $u=v+1$ in the above equation, we get $Y^{16}+Y^4+vY^4+vY+ C=0,$ which can be further simplified as $Z^4+vZ+C=0$ for $Z=Y^4+Y$. Now using Lemma~\ref{L1}, to solve $Z^4+vZ+C=0$, we have $k=2$, $m=n$ and $\alpha_{m-1}=v^{\frac{2^{2n}-1}{3}}=(B(1+B^3))^{2^{2n}-1}=1$. This gives us that $Z^4+vZ+C=0$ has no solution in $\F_{2^n}$ if $\sum_{i=0}^{n-1} v^{s_i}C^{2^{2i}} \neq 0$, where $s_{n-1}=0$, $s_{i}= \displaystyle \sum_{j=i}^{n-2} 2^{2(j+1)}$ for $0 \leq i \leq n-2$, and two solutions in $\F_{2^n}$ if the sum  $\sum_{i=0}^{n-1} v^{s_i}C^{2^{2i}}=0$, in which case, taking $c=1$, and given that $\tau^3=v$, we get $\tau=B(1+B^3)$. Then the two solutions of $Z^4+vZ+C=0$ from Lemma~\ref{L1} are $ Z_1=\sum_{i=0}^{n-1} (i+1) v^{t_i} C^{2^{2i}}$ and  $Z_2 = \sum_{i=0}^{n-1} (i+1) v^{t_i} C^{2^{2i}} + \tau$, where 
$\displaystyle t_i=\sum_{j=i}^{n-2} 2^{2(j+1)}$ for $0\leq i \leq n-1$. Observe that $t_i=s_i$ (for $r=n-1$) in our case. 

We are now left with the equation $Y^4+Y+Z_j=0$ for $j=1,2$. Again using Lemma~\ref{L1}, the equation can have either zero or two solutions depending on whether the sum $\sum_{i=0}^{n-1} (Z_j)^{2^{2i}}$ is zero or not. Notice that $\sum_{i=0}^{n-1} (Z_1)^{2^{2i}}=\sum_{i=0}^{n-1} (Z_2)^{2^{2i}}$, as the sum $\sum_{i=0}^{n-1} \tau^{2^{2i}}=0$. Thus, we have either four solutions or no solutions. This will give us our desired claim.
\end{proof}

\begin{rmk} Note that the equivalent proof for $n$ even gives 
\begin{equation*}
\nabla_F(a,b)=
\begin{cases}
2^n &~\mbox{~if~} ab=0 \mbox{~or~} a=b,\\
16 &~\mbox{~if~} ab \neq0 \mbox{~and~} a^3\neq b^3, \displaystyle \sum_{i=0}^{n/2-1} v^{s_i}C^{2^{2i}}=0, \\
&~\mbox{~and~} \displaystyle  \sum_{r=0}^{n/2-1}\left(\sum_{i=0}^{n-1} (i+1) v^{s_i} C^{2^{2i}}\right)^{2^{2r}}=0,\\
0 &~\mbox{~otherwise,~}
\end{cases}
\end{equation*}
where we have $v=\dfrac{b}{a}\left(1+\dfrac{b^3}{a^3}\right)$, $C=1+\dfrac{b^{20}+a^4b^{16}+a^{15}b^5+a^3b^{17}}{a^{16}b^4+a^{19}b}$, $s_{i}= \sum_{j=i}^{n/2-2} 2^{2(j+1)}$ for $0\leq i\leq n/2-2$ and $s_{n/2-1}=0$.
\end{rmk}

We can generalize the previous function to any cubic polynomial, but, of course, the spectrum is not going to be that explicit. For a linearized polynomial $L$, we let $\ker(L),\Im(L)$ be the kernel, respectively, the image of~$L$.
\begin{thm}
\label{cubic}
Let $F(X)=\sum_{0<i<j<n} c_{ij} X^{p^i+p^j+1}$ be a cubic function in $F_{p^n}[X]$, where $p$ is an arbitrary prime. Then, for any $a,b\in\F_{p^n}$, the second-order zero differential spectra is given by
\begin{equation*}
\nabla_F(a,b)=
\begin{cases}
p^n &~\mbox{~if~} ab=0,\\
p^{\dim (\ker L_{a,b})} &~\mbox{~if~} \delta_{a,b}\in\Im(L_{a,b}),\\ 
0 &~\mbox{~if~} \delta_{a,b}\not\in\Im(L_{a,b}),
\end{cases}
\end{equation*}
where $  \displaystyle 
L_{a,b}(X)=\sum_{0<i<j<n} c_{ij} \left( X
   \left(a^{p^i} b^{p^j}+a^{p^j} b^{p^i}\right)    +X^{p^i} \left(b a^{p^j}+a
   b^{p^j}\right)+X^{p^j} \left(b a^{p^i}+a b^{p^i}\right)\right)$, and
   $\delta_{a,b}=\sum_{0<i<j<n} c_{ij} \left((a+b)\left( a^{p^i} b^{p^j}+a^{p^j} b^{p^i}\right)  
   +a b^{p^i+p^j}+b a^{p^i+p^j}\right)$.
\end{thm}
\begin{proof}
We first compute
\allowdisplaybreaks
\begin{align*}
F(X)&=\sum_{0<i<j<n} c_{ij} X^{p^i+p^j+1}\\
F(X+a)&=\sum_{0<i<j<n} c_{ij} \left(a^{p^i+1} X^{p^j}+a^{p^j+1} X^{p^i}+X
   a^{p^i+p^j}+a^{p^i+p^j+1}\right.\\
   &\qquad \quad \left.+a^{p^i} X^{p^j+1}+a^{p^j} X^{p^i+1}+a
   X^{p^i+p^j}+X^{p^i+p^j+1} \right)\\
F(X+b)&= \sum_{0<i<j<n} c_{ij} \left(b^{p^i+1} X^{p^j}+b^{p^j+1} X^{p^i}+X
   b^{p^i+p^j}+b^{p^i+p^j+1}\right.\\
   &\qquad \quad \left.+b^{p^i} X^{p^j+1}+b^{p^j} X^{p^i+1}+b
   X^{p^i+p^j}+X^{p^i+p^j+1} \right)\\
F(X+a+b)&=\sum_{0<i<j<n} c_{ij} \left( X^{p^i+p^j+1}+X^{p^j+1} (a+b)^{p^i}+X^{p^i+1} (a+b)^{p^j}\right.\\
&\qquad\quad \left. +X^{p^j} \left( a(a+b)^{p^i}+   b(a+b)^{p^i}\right) + X^{p^i} \left(a(a+b)^{p^j}+b   (a+b)^{p^j}\right)\right.\\
&\qquad \quad \left. +X (a+b)^{p^i+p^j}+ (a+b)^{p^i+p^j+1}+X^{p^i+p^j} (a+b)
  \right).
\end{align*}
The  second-order zero differential equation $F (X + a + b) -F(X + b) -F(X + a)+ F(X) = 0$ becomes,  after simplifying where possible (a straightforward, but rather tedious computation),
\allowdisplaybreaks
\begin{align*}
0&=\sum_{0<i<j<n} c_{ij} \left(-a^{p^i+1} X^{p^j}-a^{p^j+1} X^{p^i}-X
   a^{p^i+p^j}-a^{p^i+p^j+1}+X^{p^i+1}   \left(a^{p^j}+b^{p^j}\right)\right. \\
   &\qquad \left.+X^{p^j+1}   \left(a^{p^i}+b^{p^i}\right)+a X^{p^i}
   \left(a^{p^j}+b^{p^j}\right)+a X^{p^j}   \left(a^{p^i}+b^{p^i}\right)+b X^{p^i}
   \left(a^{p^j}+b^{p^j}\right)\right.\\
  &\qquad\left. +b X^{p^j}   \left(a^{p^i}+b^{p^i}\right)+X \left(a^{p^i}+b^{p^i}\right)
 \left(a^{p^j}+b^{p^j}\right)+a \left(a^{p^i}+b^{p^i}\right) \left(a^{p^j}+b^{p^j}\right)\right.\\
   &\qquad\left. +b \left(a^{p^i}+b^{p^i}\right)
   \left(a^{p^j}+b^{p^j}\right)+a^{p^i}
   \left(-X^{p^j+1}\right)-a^{p^j} X^{p^i+1}-b^{p^j+1}
   X^{p^i}-b^{p^i+1} X^{p^j}\right.\\
   &\qquad\left. -X b^{p^i+p^j}-b^{p^i+p^j+1}-b^{p^j}
   X^{p^i+1}-b^{p^i} X^{p^j+1}\right)\\
   &=\sum_{0<i<j<n} c_{ij} \left((a+b)\left( a^{p^i} b^{p^j}+a^{p^j} b^{p^i}\right)  
   +a b^{p^i+p^j}+b a^{p^i+p^j}+X
   \left(a^{p^i} b^{p^j}+a^{p^j} b^{p^i}\right)\right.\\
   &\qquad\left. +X^{p^i} \left(b a^{p^j}+a
   b^{p^j}\right)+X^{p^j} \left(b a^{p^i}+a b^{p^i}\right)\right).
\end{align*}
If the element $\sum_{0<i<j<n} c_{ij} \left((a+b)\left( a^{p^i} b^{p^j}+a^{p^j} b^{p^i}\right)  
   +a b^{p^i+p^j}+b a^{p^i+p^j}\right)$ belongs to the image of  the affine polynomial 
   \[
   \displaystyle 
L_{a,b}(X)=\sum_{0<i<j<n} c_{ij} \left( X
   \left(a^{p^i} b^{p^j}+a^{p^j} b^{p^i}\right)    +X^{p^i} \left(b a^{p^j}+a
   b^{p^j}\right)+X^{p^j} \left(b a^{p^i}+a b^{p^i}\right)\right),
   \]
    then the  solutions for the  second-order zero differential equation  form an affine space of dimension precisely $\dim (\ker L_{a,b})$. The proof is done.
\end{proof}

\begin{rmk}
Surely, one can wonder why we have not considered DO (Dembowski-Ostrom)  polynomials, that is, polynomials of the form $F(X)=\displaystyle \sum_{0\leq i<j<n} a_{ij} X^{p^i+p^j}$  (with or without affine terms) on~$\F_{p^n}$. However, the interested reader can quickly check that the second-order zero differential equation for $F$ at $a,b$, either has no solutions or it has~$p^n$ solutions.
\end{rmk}

Budaghyan et al.~\cite{Bkrs} also considered $F(X) = X^{2^n-2^s}$ (which coincides with the inverse function $X^{-1}$ extended by $0^{-1}=0$
for $s=1$) and showed that it is $0$-APN if and only if $\gcd(n,s+1)=1$.

\begin{thm}\label{Tapn}
 Let $F(X)=X^{2^n-2^s}$ be a function over $\F_{2^n}$, where $s$ is a positive integer satisfying $\gcd(n,s+1)=1$. Then for $a,b \in \F_{2^n}$,  
\begin{equation*}
\nabla_F(a,b)=
 \begin{cases}
 2^n &~\mbox{~if~} ab=0 \mbox{~and~} a=b,\\
2^{\gcd(s,n)}-4& ~\mbox{~if~} ab \neq0, ab^{2^s}+a^{2^s}b=0 \mbox{~and~} \gcd(s,n)>1, \\
2^{n-s-\rank(E_1)}-4& ~\mbox{~if~} ab \neq0, ab^{2^s}+a^{2^s}b \neq 0, A \neq 1 \\
& \qquad\mbox{~and~}\rank (E_1) \leq (n-s)-2,\\
0 &~\mbox{~otherwise},
\end{cases}
\end{equation*}
where $A=\dfrac{ab^{2^{n-s}}+a^{2^{n-s}}b}{a^{2^{n-s}-2}(ab^2+a^2b)}$, $E_1=C_T C_T^2\cdots C_T^{2^{n-1}}-I_{n-s}$, and $C_{T}^{2^i}$ is the matrix obtained from the $(n-s)\times (n-s)$ matrix
\[
C_T=
\begin{pmatrix}
  0 & 0 &\cdots & 0 & 1+A\\
 1 & 0 & \cdots & 0 & A\\
 0 & 1 &\cdots & 0 & 0\\
  \cdots  & \cdots & \cdots & \cdots & \cdots\\
  0 & 0 &\cdots & 1 &  0
\end{pmatrix}
\]
by applying to each of its entries the automorphism $X\mapsto X^{2^i}$ and $I_k$ is the
identity matrix of order $k$.
\end{thm}
\begin{proof}
 For $a,b \in \F_{2^n}$, consider the following equation   
\begin{equation*}
F(X + a + b) +F (X + b) +F(X + a) + F (X) = 0.
\end{equation*}
If $ab = 0$ and $a=b$, then $\nabla_F(a, b) = 2^n$.
If $ab \neq 0$ and $a \neq b$, then we can write the above equation as
$$
(X + a + b)^{2^n-2^s} +(X + b)^{2^n-2^s} +(X + a)^{2^n-2^s} + X^{2^n-2^s} = 0,
$$
or equivalently,
\begin{equation}\label{e41}
 (Y + B + 1)^{2^n-2^s} +(Y + B)^{2^n-2^s} +(Y + 1)^{2^n-2^s}+ Y^{2^n-2^s} = 0,
\end{equation}
where $Y=\dfrac{X}{a}$ and $B=\dfrac{b}{a}$. Let $Y \in \{0,1,B,B+1\}$, the above equation has a solution only if $(B + 1)^{2^n-2^s} +B^{2^n-2^s} +1 = 0,$ or equivalently, $B^{2^{n-s-1}-1}=1$. However, this is possible only if $\gcd(s+1,n)>1$, and hence Equation~\eqref{e41} has no solution $Y \in \{0,1,B,B+1\}$. 

Let us assume $Y \not \in \{0,1,B,B+1\}$. Then Equation~\eqref{e41} reduced to,
\begin{equation}\label{e42}
 (B+B^{2^s})Y^{2^{s+1}}+(B+B^{2^{s+1}})Y^{2^s}+(B^{2^s}+B^{2^{s+1}})Y=0.
\end{equation}
Clearly, $Y=0, 1, B$ and $B+1$ satisfy the above equation. We have two possible cases here depending on whether $B+B^{2^s}=0$ or not.
One can see that there are $2^{\gcd(s,n)}$ possible values of B for which $B+B^{2^s}=0$. If $\gcd(s,n)=1$, then $B \in \{0,1\}$, which is not possible. Hence, we consider the case $B+B^{2^s}=0$ when $\gcd(s,n)>1$, and in this case the Equation~\eqref{e42} is reduced to $Y^{2^s}+Y=0$ which has $2^{\gcd(s,n)}$ solutions in $\F_{2^n}$.

We now assume $B+B^{2^s} \neq 0$. Observe that the sum of these coefficients is~1.
Raising Equation~\eqref{e42} to the power $2^{n-s}$, we get 
\begin{equation*}
 (B^{2^{n-s}}+B)Y^{2}+(B^{2^{n-s}}+B^{2})Y+(B+B^{2})Y^{2^{n-s}}=0,
\end{equation*}
or equivalently, dividing by $B^2+B$ and using the notation $\displaystyle A=\frac{B^{2^{n-s}}+B}{B^2+B}$,
\begin{equation}
\label{eq:lin_trin}
 Y^{2^{n-s}}+ A Y^2+(1+A)Y=0.
\end{equation}
This equation is connected to the known Blondeau-Canteaut-Charpin (BCC) conjecture~\cite{BCC} on $f_t(X)=X^{2^t-1}$ on $\F_{2^n}$, namely:
\begin{itemize}
\item[$(i)$]
$f_t$ is APN on $\F_{2^{2m}}$ if and only if $t=2$ (this was shown to be true by 
G\"olo\u glu~\cite{G12}).
\item[$(ii)$]
$f_t$ is APN on $\F_{2^{2m+1}}$ if and only if $t\in\{2,m+1,2m\}$ (still open).
\end{itemize}
Writing the differential equation for $f_t$ (for us, $t=n-s$, $\gcd(t-1,n)=\gcd(s+1,n)=1$) at some $a'$, one can derive the equation (where $Z:=1/X$),
 \begin{equation}
  \label{main:BCC}
  Z^{2^t}+A'Z^2+(1+A')Z=0, \text{ where } A'=\frac{1+a'^{2^t-1}}{1+a'}.
  \end{equation}
While we cannot prove the BCC conjecture, we can derive some conditions for the kernel of  Equation~\eqref{eq:lin_trin}.
If $A=1$, that is $B^{2^{n-s}}+B^2=0$, Equation~\eqref{eq:lin_trin} transforms into $Y^{2^{n-s}}+  Y^2=0$, that is,  $Y^{2^{n-s}-1}+Y=0$, which has $2^{\gcd(n-s-1,n)}=2^{\gcd(s+1,n)}=2$ solutions, (these are precisely $Y=0,1$) that we have already excluded.
We now let $A\neq 1$. We define  the companion matrix of the trinomial $T(Y)= Y^{2^{n-s}}+ A Y^2+(1+A)Y$ of~\eqref{eq:lin_trin} be given by the $(n-s)\times (n-s)$ matrix
\[
C_T=
\begin{pmatrix}
  0 & 0 &\cdots & 0 & 1+A\\
 1 & 0 & \cdots & 0 & A\\
 0 & 1 &\cdots & 0 & 0\\
  \cdots  & \cdots & \cdots & \cdots & \cdots\\
  0 & 0 &\cdots & 1 &  0
\end{pmatrix}.
\]
We now use McGuire and Sheekey result~\cite[Theorem 6]{MS19}(or ~\cite[Theorem 1.4]{PZ}), which, particularized  for our linearized trinomial $T$, implies that the dimension (over $\F_2$) of the kernel of $T$ (in terms of the rank of a matrix) is given by
\[
\dim_{\F_2} \ker T=n-s-\rank (E_1),
\]
where $E_1=C_T C_T^2\cdots C_T^{2^{n-1}}-I_{n-s}$, where
\[
C_T^{2^i} =\begin{pmatrix}
  0 & 0 &\cdots & 0 & (1+A)^{2^i}\\
 1 & 0 & \cdots & 0 & A^{2^i}\\
 0 & 1 &\cdots & 0 & 0\\
  \cdots  & \cdots & \cdots & \cdots & \cdots\\
  0 & 0 &\cdots & 1 &  0
\end{pmatrix}
\]
is the matrix obtained from $C_T$
by applying to each of its entries the automorphism $X\mapsto X^{2^i}$ and $I_k$ is the
identity matrix of order $k$. Therefore, the trinomial $T(Y)$ has $2^{n-s-\rank (E_1)}$ solutions in $\F_{2^n}$ or equivalently $2^{n-s-\rank (E_1)}- 4$ solutions in $\F_{2^n} \setminus \{0,1,B,B+1\}$ if $\rank (E_1) \leq (n-s)-2$. This will give us the desired claim.
\end{proof} 
To illustrate Theorem~\ref{Tapn}, we give the following example.
\begin{exmp}
 Let $F(X)=X^7$ be a function over the finite field $\F_{2^7}$, where $\F_{2^7}^{*}=<g>$. Here, $n=7$, $s=4$ and $\gcd(n,s+1)=\gcd(7,5)=1$. For $(a,b)=(g^2,g)$, we have
 \[
C_T=
\begin{pmatrix}
  0 & 0  & g^5 + g^3 + g\\
 1 & 0  & g^5 + g^3 + g + 1\\
 0 & 1  & 0\\
\end{pmatrix}
\]
and $E_1=C_T C_{T}^{2}C_{T}^{2^2}C_{T}^{2^3}C_{T}^{2^4}C_{T}^{2^5}C_{T}^{2^6}-I_3$ is a zero matrix with rank $0$. Thus, the above theorem would give us $\nabla_F(g^2,g)=2^{3}-4=4$. Notice that $\nabla_F(a,b) \leq 4$ corresponding to any pair $(a,b) \in \F_{2^n} \times \F_{2^n}$ with $ab \neq 0$ and $a \neq b$. Thus, we have that the Feistel boomerang uniformity of $F(X)=X^7$ over $\F_{2^7}$ is $4$.
\end{exmp}

\begin{rmk}
 Some particular cases of Theorem~\textup{\ref{Tapn}}  have already  been addressed, independently. For instance,  when $n-s=t=2$, $X^{2^n-2^s}$ is an APN function and for APN functions the FBCT spectra is already known in view of Lemma~\textup{\ref{L00}}. Moreover, Garg et al.~\textup{\cite{GHRS}} computed the FBCT spectra when $n=2m$ and $n-s=t=m,$ and in fact, they also showed that for $n=2m+1$ and $n-s=t=m+2$, $X^{2^t-1}$ has the F-Boomerang uniformity at most $4$. Recently, Man et al.~\textup{\cite{MMN}} computed the FBCT spectra  when $n-s=t=m+1$.
\end{rmk}

As a consequence of Theorem~\ref{Tapn}, the following corollary in a special case of $n-s=t=3$ is almost immediate.
\begin{cor}
\label{c2}
 Let $s$ and $n$ be positive integers such that $\gcd(s,n+1)=1, n-s=3$, and $F(X)=X^{2^n-2^s}$ be a function over the finite field $\F_{2^n}$. Then the Feistel boomerang uniformity of $F$ is $4$.
\end{cor}
\begin{proof}
Given the conditions on $s$ and $n$, one can easily deduce that $\gcd(s,n)$ is either 1 or 3. From Theorem~\ref{Tapn}, it is clear that for $ab \neq 0$, $a \neq b$ and $ab^{2^s}+a^{2^s}b=0$, $\nabla_F(a,b)=2^{\gcd(s,n)}-4=4$. We next consider the case when $ab \neq 0$, $a \neq b$ and $A \neq 1$, where $A=\dfrac{ab^8+a^8b}{a^6(ab^2+a^2b)}$, then we have $\nabla_F(a,b)=2^{3-\rank(E_1)}-4$ where $E_1$ is a $3 \times 3$ matrix. Since the $\rank(E_1)$ is less than or equal to $n-s-2$, it is either 0 or 1. In other words, either $\nabla_F(a,b) =4$ or $\nabla_F(a,b)=0$, respectively. Thus, the above arguments clearly indicate that  $\nabla_F(a,b) \leq 4$ for $ab \neq 0$ and $a \neq b$. Since $F$ is not an APN function (cf. \cite[Theorem 5]{BCC}), it follows from Lemma~\ref{L00} that the Feistel boomerang uniformity of $F$ is 4. 
\end{proof}

Note that Corollary~\ref{c2} gives only the Feistel boomerang uniformity of the function $X^{2^n-2^s}$ over $\F_{2^n}$ when $t=n-s=3$ and $\gcd(n,s+1)=1$. However, we adopt a slightly different approach to determine the complete Feistel boomerang spectra of this function in the following theorem. It may be noted that the differential spectrum of $F(X)=X^7$ was studied in~\cite[Theorem 5]{BCC}.

\begin{thm}
 Let $F(X)=X^{2^n-2^s}$ be a function over $\F_{2^n}$, where $s$ is a positive integer satisfying $\gcd(n,s+1)=1$ and $n-s=3$. Then for $a,b \in \F_{2^n}$, we have
 \begin{equation*}
 \nabla_F(a,b)=
\begin{cases}
2^n &~\mbox{~if~} ab=0 \mbox{~or~} a=b,\\
4 &~\mbox{~if~} ab \neq0, a\neq b, ab^{2^s}+a^{2^s}b=0\mbox{~and~}\gcd(s,n)=3 ,\\
&~\mbox{~or~} ab \neq0, a\neq b,  \Tr\left(\dfrac{a^7(ab^2+a^2b)}{a^2b^8+a^8b^2}\right)=1,\\
&~\Tr\left(\dfrac{a^3}{a^2b+ab^2}\right)=1; t_1, t_2 ~\mbox{~are cubes in~} \F_{2^{2n}}, \\
&~\mbox{~and~} \Tr(Z_1)=\Tr(Z_2)=0, \\
0 &~\mbox{~otherwise,}\\
\end{cases}
\end{equation*}
where $Z_1,Z_2$ are solutions of $a^6Z^2+a^3(ab^2+a^2b)Z+a^2b^4+a^4b^2+a^6=0$ and $t_1,t_2$ are solutions of $a^7(ab^2+a^2b)t^2+(a^8b^2+a^2b^8)t+a^7(ab^2+a^2b)=0$. Moreover, the Feistel boomerang uniformity of $F$ is $4$.
\end{thm}
\begin{proof}
It is easy to observe that $n-s=3$ and $\gcd(s+1,n)=1$ would together imply that $n$ is odd. Consider the following equation 
\begin{equation*}
F(X + a + b) +F (X + b) +F(X + a) + F (X) = 0.
\end{equation*}
If $ab = 0$ and $a=b$, then $\nabla_F(a, b) = 2^n$.
If $ab \neq 0$ and $a \neq b$, then we can write the above equation as

\begin{equation}\label{eq51}
 (Y + B + 1)^{2^n-2^s} +(Y + B)^{2^n-2^s} +(Y + 1)^{2^n-2^s}+ Y^{2^n-2^s} = 0,
\end{equation}
where $Y=\dfrac{X}{a}$ and $B=\dfrac{b}{a}$. We know from the proof of the Theorem~\ref{Tapn} that Equation~\eqref{eq51} has no solution in $\{0,1,B,B+1\}$. Let us assume $Y \not \in \{0,1,B,B+1\}$, then one can write the Equation~\eqref{eq51} as

\begin{equation}\label{e52}
 (B+B^{2^s})Y^{2^{s+1}}+(B+B^{2^{s+1}})Y^{2^s}+(B^{2^s}+B^{2^{s+1}})Y=0.
\end{equation}
Clearly, $Y=0, 1, B$ and $B+1$ satisfy the above equation. We first consider the case when $B+B^{2^s}=0$. Observe that there are $2^{\gcd(s,n)}$ such $B \in \F_{2^n}$ satisfying $B+B^{2^s}=0$. Notice that $\gcd(s,n)$ is either 1 or 3. Thus, when $\gcd(s,n)=1$, the equation $B+B^{2^s}=0$ has only two solutions, namely $B=0$ and $B=1$, which we have already considered. Hence, we consider the case $B+B^{2^s}=0$ when $\gcd(s,n)=3$ and in this case, the Equation~\eqref{e52} has exactly $4$ solutions in $\F_{2^n}\setminus\{0,1,B,B+1\}$. Next, we consider $B+B^{2^s} \neq 0$, and in this case, we can reduce Equation~\eqref{e52} as given below
\begin{equation}\label{e53}
 Y^{2^{3}}+ A Y^2+(1+A)Y=0,
\end{equation}
where $A=\dfrac{B^{2^{3}}+B}{B^2+B}$. We can further rewrite the  Equation~\eqref{e52} as
$$
Y^{2^3}+Y^{2^2}+Y^{2^2}+Y^{2}+Y^{2}+AY^2+Y+AY=0,
$$
or equivalently,
$
(Y^{2}+Y)^{2^2}+(Y^{2}+Y)^{2}+(1+A)(Y^{2}+Y)=0.
$
Substitute $Z=Y^2+Y$ in the above equation to get 
\begin{equation}\label{e54}
 Z^{2^2}+Z^{2}+(1+A)Z=0.
\end{equation}
Clearly, $Z=0$ and $Z=B^2+B$ are solutions of Equation ~\eqref{e54}. Then the cubic equation $Z^3+Z+(1+A)=0$ would have at least one solution in $\F_{2^n}$, or more precisely, the cubic equation can have either one or three solutions in $\F_{2^n}$. From Lemma~\ref{L4}, if $\Tr\left(\frac{1}{1+A}\right)\neq 1$, then $Z^3+Z+(1+A)=0$ has only one solution, that is, $Z=B^2+B$. This would further imply that $Y \in \{0,1,B,B+1\}$, which is a contradiction to our assumption. Next, let the cubic equation has three solutions in $\F_{2^n}$. Now, Lemma~\ref{L4} ensures that this is possible only when $\Tr\left(\frac{1}{1+A}\right)=1$ and $t_1$, $t_2$ are cubes in $\F_{2^{2n}}$, where $t_1$, $t_2$ are roots of $t^2 + (1+A)t + 1 = 0$. Let $u_1$ and $u_2$ be the two solutions of $Z^3+Z+(1+A)=0$, except $B^2+B$, then they must satisfy the quadratic equation $Z^2+(B^2+B)Z+B^4+B^2+1=0$. One can see from Lemma~\ref{L3} that this quadratic equation would have no solution in $\F_{2^n}$ if $\Tr\left(\frac{1}{B^2+B}\right)\neq 1$ and exactly two solutions in $\F_{2^n}$ if $\Tr\left(\frac{1}{B^2+B}\right) = 1$. Let the two solutions in later case be $Z_1$ and $Z_2$, then it is easy to observe that $\Tr(Z_1)=\Tr(Z_2)$. Now, we look at the equations $Y^2+Y+Z_i=0$, for $i=1,2$, and  in view of Lemma~\ref{L3}, we get that the original Equation~\eqref{e53} has four solutions in $\F_{2^n}$ when $\Tr(Z_1)=\Tr(Z_2)=0$,  and otherwise has no solution. This completes the proof of the theorem.
\end{proof}

\section{Further comments}
\label{comments}

In this paper we concentrate on the second-order zero differential uniformity and spectra, connected to the (coefficients of the) Feistel Boomerang Connectivity Table (FBCT), the counterpart concept for Feistel Network-based ciphers (like $3$-DES, CLEFIA, etc.), as introduced by~\cite{Bouk}. Recall that it is known that a function is APN in even characteristic, and PN in odd characteristic if and only if the  second-order zero differential uniformity is zero. It is natural to inquire what happens for APN functions in odd characteristic.
We start with a rather intriguing result, where we show that if an odd or even function $F$ is second-order zero differentially $1$-uniform then it has to be  an APN function. We then proceed to consider several APN, or other low differential uniform functions and check their Feistel boomerang uniformity, or even compute, for most of them, their FBCT entries. This sheds further light on the  constraints between the differential uniformity and the  second-order zero differential uniformity and  the behavior of the latter concept. Surely, there is much to be done in this direction, and we point out that perhaps some good bounds relating the differential uniformity and the second-order zero differential uniformity would be an avenue of research, as well as investigate how various other classes of functions (besides the ones we considered) behave under this recent concept.


\begin{thebibliography}{99}
 
\bibitem{Biham91} E. Biham, A. Shamir, {\it Differential cryptanalysis of DES-like cryptosystems,} J. Cryptol. 4:1 (1991), 3--72.

\bibitem{BCC} C. Blondeau, A. Canteaut, P. Charpin, {\it Differential properties of $X \rightarrow X^{2^t-1}$,} IEEE Trans. Inf. Theory 57:12 (2011), 8127--8137.

\bibitem{Bouk} H. Boukerrou, P. Huynh, V. Lallemand, B. Mandal, M. Minier, {\it On the Feistel counterpart of the boomerang connectivity table,} IACR Trans. Symmetric Cryptol. 1 (2020), 331--362.

\bibitem{BoCa} C. Boura, A. Canteaut, {\it On the boomerang uniformity of cryptographic S-boxes,} IACR Trans. Symmetric Cryptol. 3 (2018), 290--310.

\bibitem{Bksrs} L. Budaghyan, N. Kaleyski, S. Kwon, C. Riera, P. St\u anic\u a, {\it Partially APN Boolean functions and classes of functions that are not APN infinitely often} Cryptogr.
Commun. 12:3 (2020), 527--545.

\bibitem{Bkrs} L. Budaghyan, N. Kaleyski, C. Riera, P. St\u anic\u a, {\it Partially APN functions with APN-like polynomial representations,} Des. Codes Crypt. 88 (2020), 1159--1177.

\bibitem{BP} L. Budaghyan, M. Pal {\it Arithmetization oriented APN functions}, Cryptology ePrint Archive, https://eprint.iacr.org/2023/1081 (2023).

\bibitem{cid} C. Cid, T. Huang, T. Peyrin, Y. Sasaki,  L. Song, {\it Boomerang connectivity table: a new cryptanalysis tool.} In: J. Nielsen, V. Rijmen (ed) Advances in Cryptology-EUROCRYPT'18. LNCS 10821, Springer, Cham, 683--714 (2018).

\bibitem{RH} R. S. Coulter, M. Henderson, {\it A note on the roots of trinomials over a finite field,} Bull. Austral. Math. Soc. 69 (2004), 429--432.

\bibitem{Dickson06}
L. E. Dickson,
{\em Criteria for the irreducibility of functions in a finite field}, 
Bulletin of the American Math. Soc. 13 (1906), 1--8.

\bibitem{DHKM} H. Dobbertin, T. Helleseth, P. V. Kumar, and H. Martinsen, {\it Ternary m-sequences with three-valued crosscorrelation function: New decimations of Welch and Niho type,} IEEE Trans. Inf. Theory 47:4 (2001), 1473--1481.

\bibitem {EM} S. Eddahmani, S. Mesnager, {\it Explicit values of the DDT, the BCT, the FBCT, and the FBDT of the inverse, the Gold, and the Bracken-Leander S-boxes}, Cryptogr. Commun. 14 (2022), 1301--1344.

\bibitem{EFRS} P. Ellingsen, P. Felke, C. Riera, P. St\u anic\u a, A. Tkachenko, {\it $c$-differentials, multiplicative uniformity, and (almost) perfect $c$-nonlinearity,} IEEE Trans. Inf. Theory 66(6) (2020), 5781--5789.

\bibitem{GHRS} K. Garg, S.U. Hasan, C. Riera, P. St\u anic\u a, {\it The second-order zero differential spectra of some functions over finite fields}, arXiv preprint arXiv:2309.04219 (2023).

\bibitem{G12}
F. G\"olo\u glu, {\em A note on the differential properties of $x\mapsto x^{2^t-1}$}, IEEE Tran. Inf. Theory 58:11 (2012), 6986--6988.

\bibitem{HRS99} T. Helleseth, C. Rong, D. Sandberg, {\it New families of almost perfect nonlinear power functions}. IEEE Trans. Inf. Theory \textbf{45} (1999), 475--485.

\bibitem{LYT} X. Li, Q. Yue, D. Tang, {\it The second-order zero differential spectra of almost perfect nonlinear functions and the inverse function in odd characteristic}, Cryptogr. Commun. 14.3 (2022), 653--662.

\bibitem{MMN} Y. Man, S. Mesnager, N. Li, X. Zeng, X. Tang, {\it In-depth analysis of S-boxes over binary finite fields concerning their differential and Feistel boomerang differential uniformities}, arXiv preprint arXiv:2309.01881 (2023).

\bibitem{MS19}
G. McGuire, J. Sheekey, {\em A Characterization of the Number
of Roots of Linearized and Projective Polynomials in the Field of Coefficients},
Finite Fields Appl. 57 (2019), 68--91.

\bibitem{Nyberg} K. Nyberg, {\it Differentially uniform mappings for cryptography,} In T. Helleseth (ed), Advances in Cryptology-EUROCRYPT’93, LNCS 765, Springer, Heidelberg, 55--64 (1994).

\bibitem{PZ} O. Polverino, Z. Ferdinando, {\it On the number of roots of some linearized polynomials,} Linear Algebra Appl. 601 (2020),  189--218.

\bibitem{Loops19}
A. Pott,
 {\em Partially almost perfect nonlinear permutations}, In LOOPS, Budapest, Hungary, 2019.

\bibitem{SSA} T. Shirai, K. Shibutani, T.  Akishita, S. Moriai, T.  Iwata, {\it The $128$-bit blockcipher CLEFIA (extended abstract,)} In: Biryukov A. (eds) Fast Software Encryption-FSE 2007. LNCS 4593, Springer, Berlin, Heidelberg, 181--195 (2007).

\bibitem{DW} D. Wagner, {\it The boomerang attack,} In: L. R. Knudsen (ed.) Fast Software Encryption-FSE 1999. LNCS 1636, Springer, Berlin, Heidelberg, 156--170 (1999). 

\bibitem{William} K.S. Williams, {\it Note on cubics over $\F_{2^n}$ and $\F_{3^n}$} J. Number Theory 7(4) (1975),   361--365.

\end{thebibliography}
\end{document}